\documentclass[11pt]{llncs}
\advance\hoffset by -1.5cm
\advance\voffset by -0.5cm
\usepackage{amssymb,amsfonts,amsmath}

\renewcommand{\leq}{\leqslant}
\renewcommand{\geq}{\geqslant}
\newcommand{\NN}{\mathbb{N}}

\newcommand{\TT}{\mathcal{T}}

\newcommand{\median}{\tilde{k}}
\newcommand{\percentilp}{\ensuremath{x_p}}

\begin{document}

\title{The median of the distance between two leaves in a phylogenetic tree}

\author{Arnau Mir\inst{1} \and Francesc Rossell\'o\inst{1}}

\institute{Department of Mathematics and Computer Science, University of the
Balearic Islands, E-07122 Palma de Mallorca,
\texttt{\{arnau.mir,cesc.rossello\}@uib.es}}

\maketitle

\begin{abstract}
We establish a limit formula for the median of the distance between two leaves in a fully resolved unrooted phylogenetic tree with $n$ leaves. More precisely, we prove that this median is equal, in the limit, to $\sqrt{4\ln(2) n}$.
\end{abstract}

\section{Introduction}
The definition and study of metrics for the comparison of phylogenetic trees is a classical problem in phylogenetics \cite[Ch.~30]{Fel04}, motivated, among other applications, by the need to compare alternative phylogenies for a given set of organisms obtained from different
datasets  or using  different methods.  Many metrics for the
comparison of rooted or unrooted phylogenetic trees on the same set of taxa have been proposed so far. Some of the most popular such metrics are based
on the comparison of the vectors of distances between pairs of
taxa in the corresponding trees.
But, in contrast with other metrics, the statistical properties of these metrics  are mostly unknown.

 Steel and Penny \cite{StePen93} computed the mean  value of the square of the metric  for  fully resolved unrooted trees defined through the euclidean distance between their vectors of distances (they called it the \emph{path difference metric}).
One of the main ingredients in their work was the explicit computation of the mean value  and the variance  of the distance $d$ between two leaves in a fully resolved  unrooted phylogenetic tree with $n$ leaves, obtaining that
 $$
\mu(d)=\frac{2^{2(n-2)}}{\binom{2(n-2)}{n-2}}\sim \sqrt{\pi n},\quad
\mbox{Var}(d)=4n-6-\mu(d)-\mu(d)^2
$$

In this work we continue the statistical analysis of this random variable $d$, by giving an expression for its median that allows the derivation of a limit formula for it. We hope our result will constitute a first step towards obtaining a formula for the median of the aforementioned squared path difference metric between fully resolved  unrooted phylogenetic trees, a problem that still remains open.

\section{Preliminaries}

In this paper, by a  \emph{phylogenetic tree} on a set $S$  we mean a \emph{fully resolved}  (that is, with   all its internal nodes of degree 3) unrooted tree  with its leaves bijectively labeled in  the set $S$.   Although in practice $S$ may be any set of taxa, to fix ideas  we shall always take $S=\{1,\ldots,n\}$, with $n$ the number of leaves of the tree, and we shall use the term \emph{phylogenetic tree with $n$ leaves} to refer to a phylogenetic tree on this set.  For simplicity, we shall always identify a leaf of a phylogenetic tree with its label. 
 
Let $\TT_n$ be the set of (isomorphism classes of) phylogenetic trees with $n$ leaves. It is well known \cite{Fel04} that $|\TT_1|=|\TT_2|=1$ and
$|\TT_n|=(2n-5)!!=(2n-5)(2n-7)\cdots 3 \cdot 1$, for every $n\geq 3$.

\section{Main result}

Let $k,l\in S=\{1,\ldots,n\}$ be any two different labels of trees in $\TT_n$. 
 The \emph{distance} $d_T(k,l)$ between the leaves $k$ and $l$ in a phylogenetic tree $T\in \TT_n$ is the length of the unique path between them.
Let's consider the random variable 
$$
d_{kl}=\mbox{distance between the labels $k$ and $l$ in one tree in ${\TT}_n$.}
$$
The possible values of $d_{kl}$ are $1,2,\ldots,n-1$. 

Our goal is to estimate the value  $\mbox{median}(n)$ of the median of this variable~$d_{kl}$ on $\TT_n$ when the tree and the leaves are chosen equiprobably.  In this case, $d_{kl} =d_{12}$, and thus we can reduce our problem to compute the median of the variable~$d:=d_{12}$.

For every $i=1,\ldots, n-1$, let $c_i$ be the cardinal of $\{T\in {\TT}_n\mid d_T(1,2)=i\}$. Arguing as in  \cite[p.~140]{StePen93}, we have the following result.

\begin{lemma}\label{lem1}
$c_{n-1}=(n-2)!$ and, for every $i=1,\ldots,n-2$,
$$
c_i = (n-2)! \dfrac{(i-1) (n-1)\cdots (2n-i-4)}{(2(n-i-1))!!}=\frac{(i-1)(2n-i-4)!}{(2(n-i-1))!!}.
$$
\end{lemma}

\begin{proof}
Consider  the  function
$B(x)=1-\sqrt{1-2x}.$  
By \cite[p.~140]{StePen93}, we have that
$c_i =\frac{\partial^{n-2}}{\partial x^{n-2}} \left( B(x)^{i-1}\right)\vrule_{x=0}.$
Using that
$$
B(x)^{i-1}\hspace*{-2ex}=x^{i-1}+\frac{i-1}{2} x^i+\frac{(i-1) (i+2)}{8} x^{i+1}+\cdots +
\frac{(i-1) (i+l)\ldots (i+2l-2)}{(2l)!!} x^{i-1+l}+\ldots,
$$
we obtain the formulas in the statement.\qed
\end{proof}

\begin{lemma}\label{lem2}
For every $k=1,\ldots,n-1$,
$\displaystyle \frac{1}{(2n-5)!!}\sum_{i=1}^k c_i =1-\frac{2^k (n-3)! (-k+2 n-4)!}{2 (2 n-5)! (-k+n-2)!}.$
\end{lemma}

\begin{proof}
Taking into account that  $(2j)!!=2^j j!$ and $(2j+1)!!=\frac{(2j+1)!}{2^j j!}$, for every $j\in \NN$, and using Lemma \ref{lem1}, we have:
\begin{eqnarray*}
\frac{1}{(2n-5)!!}\sum_{i=1}^k c_i & = & \frac{(n-3)!}{4 (2n-5)!} \sum_{i=2}^k \frac{(i-1) 2^i (2n-i-4)! }{(n-i-1)!}
\\ & = & \frac{(n-3)!}{4 (2n-5)!} \sum_{i=1}^{k-1} \frac{i 2^{i+1} (2n-i-5)! }{(n-i-2)!}.
\end{eqnarray*}
We use now the method in \cite[Chap.~5]{PetWilZei96}  to compute $S_k = \sum_{i=1}^{k-1} \frac{i 2^{i+1} (2n-i-5)! }{(n-i-2)!}$. 

Set $t_i =  i 2^{i+1} (2n-i-5)! /(n-i-2)!$. Then
$$
\frac{t_{i+1}}{t_i}= \frac{2 (1+i) (2+i-n)}{i (5+i-2n)}.
$$ 
The next step is to find three polynomials $a(i), b(i)$ and $c(i)$ such that  
$$
\frac{t_{i+1}}{t_i}=\frac{a(i)}{b(i)}\cdot \frac{c(i+1)}{c(i)}.
$$
We take $a(i)=2 (2+i-n),$ $b(i)=5+i-2n$ and $c(i)=i$. Next, we have to find a polynomial $x(i)$ such that
$a(i) x(i+1)-b(i-1) x(i)=c(i)$.
The polynomial $x(i)=1$ satisfies this equation. Then, by \cite[Chap.~5]{PetWilZei96},
$$S_k = \frac{b(k-1) x(k)}{c(k)} t_k +g(n)=\frac{(4+k-2n) 2^{k+1} (2n-k-5)!}{(n-k-2)!}+g(n),$$ 
where $g$ is a function of $n$. We find this function from the case $k=2$:
$$
 \frac{4 (2 n-6)!}{(n-3)!}=S_2 =\frac{8 (6-2 n) (2 n-7)!}{(n-4)!}+g(n).
$$
From this equality we deduce that
$g(n)=\dfrac{4 (2 n-5)!}{(n-3)!}$.
We conclude that:
$$
S_k = \sum_{i=1}^{k-1} \frac{i 2^{i+1} (2n-i-5)! }{(n-i-2)!}=\frac{(4+k-2n) 2^{k+1} (2n-k-5)!}{(n-k-2)!} + \frac{4 (2 n-5)!}{(n-3)!}.
$$
The formula in the statement follows from this expression. \qed
\end{proof}

\begin{theorem}
$\dfrac{\mathrm{median}(n)}{\sqrt{4\ln(2)n}}  = 1+O\left(n^{-1/2}\right)$. In particular, $\displaystyle \lim_{n\to\infty} \frac{\mathrm{median}(n)}{\sqrt{4\ln(2)n}}=1.$
\end{theorem}

\begin{proof}
To simplify the notations, we shall denote $\mathrm{median}(n)$ by $\median$. By definition,
$$
\median = \max\Bigl\{k\in \NN\mid \sum_{i=1}^k c_i \leq \frac{|\TT_n|}{2} \Bigr\}\!
\! = \max \Bigl\{k\in \NN \mid \frac{2^k (n-3)! (-k+2 n-4)!}{2 (2 n-5)! (-k+n-2)!} \geq \frac{1}{2} \Bigr\}.
$$
Thus, $\median$ is the largest integer value such that
$$
2^{\median} (n-3)! (-\median +2 n-4)! \geq (2 n-5)! (-\median+n-2)!.
$$
If we simplify this inequation and take logarithms, this condition becomes
\begin{equation}
\median\ln(2)\geq \sum_{j=3}^{\median +1} \ln \left(\frac{2n-(j+2)}{n-j} \right)=\sum_{j=3}^{\median +1} \ln \left(\frac{2-\frac{j+2}{n}}{1-\frac{j}{n}} \right).
\label{EQMEDIAN}
\end{equation}
Combining the development of the function $\ln(\frac{2-(j+2){x}}{1-{j}{x}})$ in $x=0$,
$$
\ln \left(\frac{2-(j+2){x}}{1-{j}{x}}\right)=\ln(2)+\frac{1}{2} (j-2) x+\frac{1}{8} (j-2) (3 j+2) x^2+O\left(x^3\right),
$$
with equation~(\ref{EQMEDIAN}), we obtain: 
$$
\ln(2) \geq \frac{1}{2n}\sum_{j=3}^{\median +1} (j-2) + O\left(\frac{\median^3}{n^2}\right)=\frac{\median (\median -1)}{4 n}+ O\left(\frac{\median^3}{n^2}\right).
$$
So, the first order term of the median $\median$ will be the largest integer value that satisfies
${\median^2}/{4n}\leq\ln(2)$.
Therefore, the median will be the closest integer to $\sqrt{4\ln(2)n}$, from where the thesis in the statement follows.\qed
\end{proof}

\section{Conclusions}

We have obtained a limit formula for the median of the distance between two leaves in a fully resolved unrooted phylogenetic tree with $n$ leaves. Our method allows to find more terms of the development of the median. For instance, it can be proved that
$\mbox{median}(n)\approx \sqrt{4n\ln 2}+(\frac{1}{2}-\ln 2)$.

The limit formula obtained in this work can be generalized to the $p$-percentile 
$\percentilp = \max\Big\{k\in \NN\mid \sum_{i=1}^k c_i \leq {|\TT_n|}p\Big\}$.
Indeed, using our method we obtain that
$\percentilp \approx \sqrt{-4\ln (1-p) n}$.


\begin{thebibliography}{1}
\providecommand{\url}[1]{{#1}}
\providecommand{\urlprefix}{URL }
\expandafter\ifx\csname urlstyle\endcsname\relax
  \providecommand{\doi}[1]{DOI~\discretionary{}{}{}#1}\else
  \providecommand{\doi}{DOI~\discretionary{}{}{}\begingroup
  \urlstyle{rm}\Url}\fi

\bibitem{Fel04}
Felsenstein, J.: Inferring Phylogenies.
\newblock Sinauer Associates Inc. (2004)

\bibitem{PetWilZei96}
Petkovsek, M., Wilf, H., Zeilberger, D.: $A=B$.
\newblock AK Peters Ltd. (1996).
\newblock Available on line at
  \url{http://www.math.upenn.edu/\~{}wilf/AeqB.html}

\bibitem{StePen93}
Steel, M.A., Penny, D.: Distributions of tree comparison metrics---some new
  results.
\newblock Syst. Biol. \textbf{41}, 126--141 (1993)

\end{thebibliography}
\end{document}